\documentclass[10pt, conference, letterpaper]{IEEEtran}
\IEEEoverridecommandlockouts
\usepackage[T1]{fontenc}
\usepackage{graphicx}
\usepackage{subcaption}
\usepackage{graphicx}
\usepackage{amsmath,amssymb,amsfonts,amsthm}
\usepackage{hyperref}
\usepackage{mathtools}
\usepackage[dvipsnames]{xcolor}
\usepackage{cite}

\newtheorem{problem}{Problem}
\newtheorem{theorem}{Theorem}

\newcommand{\distg}[2]{D_{#1}(#2)}

\newcommand{\hlcolor}{Yellow!35}
\newcommand{\hlcolorTwo}{LimeGreen!35}
\makeatletter
\newenvironment{btHighlight}[1][]
{\begingroup\tikzset{bt@Highlight@par/.style={#1}}\begin{lrbox}{\@tempboxa}}
{\end{lrbox}\bt@HL@box[bt@Highlight@par]{\@tempboxa}\endgroup}

\newcommand\btHL[1][]{%
  \begin{btHighlight}[#1]\bgroup\aftergroup\bt@HL@endenv%
}
\def\bt@HL@endenv{%
  \end{btHighlight}%
  \egroup
}
\newcommand{\bt@HL@box}[2][]{%
  \tikz[#1]{%
    \pgfpathrectangle{\pgfpoint{1pt}{0pt}}{\pgfpoint{\wd #2}{\ht #2}}%
    \pgfusepath{use as bounding box}%
    \node[anchor=base west, fill=\hlcolor,outer sep=0pt,inner xsep=1pt, inner ysep=0pt, rounded corners=2pt, minimum height=\ht\strutbox+2pt,#1]{\raisebox{1pt}{\strut}\strut\usebox{#2}};
  }%
}

\newenvironment{btHighlightTwo}[1][]
{\begingroup\tikzset{bt@HighlightTwo@par/.style={#1}}\begin{lrbox}{\@tempboxa}}
{\end{lrbox}\bt@HLTwo@box[bt@HighlightTwo@par]{\@tempboxa}\endgroup}

\newcommand\btHLTwo[1][]{%
  \begin{btHighlightTwo}[#1]\bgroup\aftergroup\bt@HLTwo@endenv%
}
\def\bt@HLTwo@endenv{%
  \end{btHighlightTwo}%
  \egroup
}
\newcommand{\bt@HLTwo@box}[2][]{%
  \tikz[#1]{%
    \pgfpathrectangle{\pgfpoint{1pt}{0pt}}{\pgfpoint{\wd #2}{\ht #2}}%
    \pgfusepath{use as bounding box}%
    \node[anchor=base west, fill=\hlcolorTwo,outer sep=0pt,inner xsep=1pt, inner ysep=0pt, rounded corners=2pt, minimum height=\ht\strutbox+2pt,#1]{\raisebox{1pt}{\strut}\strut\usebox{#2}};
  }%
}

\usepackage{tikz}

\definecolor{codepurple}{rgb}{0.58,0,0.82}
\usepackage{listings}
\def\BibTeX{{\rm B\kern-.05em{\sc i\kern-.025em b}\kern-.08em
    T\kern-.1667em\lower.7ex\hbox{E}\kern-.125emX}}

\makeatletter 
\newcommand{\linebreakand}{%
  \end{@IEEEauthorhalign}
  \hfill\mbox{}\par
  \mbox{}\hfill\begin{@IEEEauthorhalign}
}
\makeatother 

\begin{document}
\title{In the Search of Optimal Tree Networks: Hardness~and Heuristics}

\author{
\IEEEauthorblockN{Maxim Buzdalov}
\IEEEauthorblockA{
\textit{Aberystwyth University}\\
Aberystwyth, UK \\
mab168@aber.ac.uk}
\and
\IEEEauthorblockN{Pavel Martynov}
\IEEEauthorblockA{
\textit{ITMO University}\\
Saint-Petersburg, Russia \\
covariancemomentum@gmail.com}
\and
\IEEEauthorblockN{Sergey Pankratov}
\IEEEauthorblockA{
\textit{ITMO University}\\
Saint-Petersburg, Russia \\
zergey.gad@gmail.com}

\linebreakand

\IEEEauthorblockN{Vitaly Aksenov}
\IEEEauthorblockA{
\textit{City, University of London}\\
London, UK \\
aksenov.vitaly@gmail.com}
\and
\IEEEauthorblockN{Stefan Schmid}
\IEEEauthorblockA{
\textit{TU Berlin}\\
Berlin, Germany \\
stefan.schmid@tu-berlin.de}
}

\maketitle              
\begin{abstract}

Demand-aware communication networks are networks whose topology is optimized toward the traffic they need to serve. 
These networks have recently been enabled by novel optical communication technologies and are investigated intensively in the context of datacenters.

In this work, we consider networks with one of the most common topologies~--- a binary tree.
%
We show that finding an optimal demand-aware binary tree network is NP-hard.
Then, we propose optimization algorithms that generate efficient binary tree networks on real-life and synthetic workloads.
%
\end{abstract}

\begin{IEEEkeywords}
demand-aware networks, binary trees, NP-hardness, heuristics
\end{IEEEkeywords}

\section{Introduction}

Modern datacenters serve huge amounts of communication traffic which impose stringent performance requirements on the underlying network.
Most of the datacenters nowadays are designed for uniform (all-to-all) traffic independently of the actual traffic patterns they serve, typically, a fat-tree topology~\cite{leiserson1985fat}.

This paper explores an alternative design which has recently received attention: \emph{demand-aware} networks, that is, networks whose topology is optimized toward the traffic.
%
%
%
We focus on the most fundamental class of network topologies in this paper: the binary tree topology.
In other words, we want to design an optimal static binary tree network given a demand matrix (which can also be represented as a demand graph).

Demand-aware binary tree networks have been studied before in the literature. More specifically, researchers have studied binary \emph{search tree} networks~\cite{7066977}, and showed that such networks can be computed in polynomial time. In this paper, we relax the requirement that the tree must support the search property.
By alleviating this constraint we allow ourselves to access more optimal solutions.

In this work, we show that unlike the search tree variant, the problem of finding an optimal binary tree is NP-complete.
We propose optimization algorithms that work well on synthetic and real loads: on most of them, our generated binary tree networks outperform the binary search tree networks significantly.
All our algorithms start in some configuration and then apply newly introduced heuristics.
We present several algorithms to get the initial configuration, for example, the optimal binary search tree algorithm and a novel maximum spanning tree algorithm.
Then, we apply heuristics (namely mutations for the $(1+1)$-evolutionary algorithm) on top of these solutions: 1)~swapping the neighbours of two adjacent vertices; 2)~replacing an edge; and 3)~swapping two random subtrees.
These heuristics are used as mutations in our evolutionary algorithm which allows us to employ the power of the simulated annealing algorithm.
On average our optimization approach with random heuristics reaches $10\%$ improvement against the proposed initialization algorithms.

\noindent\textbf{Related work. }
%
%
Traditionally, datacenter designs are demand-oblivious and optimized for the ``worst-case'', often providing nearly full bisection bandwidth, catering to dense, all-to-all communication patterns.
The recent advances in optical technologies~\cite{farrington2010helios,hamedazimi2014firefly,liu2014circuit,ghobadi2016projector} have enabled easy reconfiguration of physical network topologies, leading to \emph{reconfigurable} networks.
See~\cite{avin2018demandaware} for an algorithmic taxonomy of the field.

Many existing network design algorithms rely on estimates or snapshots of the traffic demands, from which an optimized network topology is (re)computed periodically~\cite{singla2010proteus, avin2022demand, avin2020demand, avin2018rdan, foerster2018characterizing}.
However, all these algorithm are intended for more general networks and only provide approximations lacking the optimality property: in-between periods they build non-optimal networks for the given demand.
%

We are aware of only two results providing optimal demand-aware networks.
The first one~\cite{GareyJS76} shows that the problem is NP-complete when a demand-aware network should be a graph of degree two, i.e., a line or a cycle.
Note that it does not follow that the problem for a binary tree topology is NP-complete since the restrictions on networks are different.
For example, it was shown in~\cite{7066977} that a problem for a binary search tree topology can be solved in polynomial time.
Each network with this topology has the following property: all nodes in the left subtree have smaller identifiers than the root, while those in the right subtree have larger identifiers.
The construction algorithm is simply dynamic programming on the segments.


\noindent\textbf{Roadmap. }
In Section~\ref{sec:background}, we state the problem and discuss optimization algorithms. In Section~\ref{sec:proof}, we prove that the problem is NP-hard. In Section~\ref{sec:heuristics}, we propose our heuristic algorithms. In Section~\ref{sec:exp}, we show that these algorithms find good enough solutions. We conclude with Section~\ref{sec:conclusion}.

\section{Background}
\label{sec:background}

\subsection{Demand-aware network design problem}




\paragraph{Demand Matrix}

The \emph{demand} of a network refers to the pattern of usage and traffic on the network over a certain period of time. It is a characterization of the amount and type of data that is transmitted across the network at different times of the day, week, or month.

The demand is affected by various factors such as the type of applications being used, the time of day, etc.
The nature of the demand for a network can vary widely depending on the specific network and its usage patterns~\cite{avin2020complexity,roy2015inside,doe2016characterization,alizadeh2013pfabric}.
%

In our chosen theoretical model, \emph{demand}, or \emph{load}, can be defined as a square \emph{demand matrix} $W$ of size $n$ where $n$ is the number of hosts in the network. Value $W_{ij}$ numerically denotes the amount of traffic between nodes $i$ and $j$ of the network~--- we can think of it as the frequency of communication between these two nodes (or the probability).
In this paper, we abstract from the actual meaning of these numbers.


\paragraph{Static Optimal Networks}
A static optimal network is a network that is designed to provide the best possible performance and efficiency for a particular set of conditions, without considering changes in traffic patterns or usage over time.

In our theoretical model, we want to obtain the binary tree network which minimizes $\sum_{1 \leq i,j \leq n} W_{ij} \cdot D_{ij}$ where $W$ is the demand matrix, i.e., an adjacency matrix of a demand graph, and $D$ is a distance matrix in the constructed network.
We refer to the value of the sum above as the \emph{cost} of the network: 
\begin{equation}
C(D, W) = \sum\limits_{1 \leq i,j \leq n} W_{ij} \cdot D_{ij}. \label{eq:cost}
\end{equation}

\paragraph{Topological limitations}
It is obvious that the optimal solution without any limitations is the full graph.
However, this topology is not scalable.

Thus, it seems natural to narrow down the possible network topology for our statically optimal network. There exist several standard topologies: a line; a binary tree; a tree topology; and $\Delta$-bounded topology, more specifically, $3$-bounded topology.

\paragraph{Our problem}

In this work, we are interested in demand-aware networks with the binary tree topology.
%
%
However, in a simple binary tree we do not require an additional search property.
We show that a problem to find an optimal binary tree topology is NP-complete and then we present several optimization algorithms that achieve better results than the optimal binary search tree.

\begin{problem}[Optimal Binary Tree Problem (OBT)]
Let $W = [w_{ij}]$ be a symmetric demand matrix.
Let $C$ be the required arrangement cost.
For a connected undirected graph $G$, $D_G(i, j)$ denotes the shortest distance between vertices $i$ and $j$.

Question: Does there exist a binary tree graph $B$, such that:

\[\sum_{i > j} w_{ij} \cdot D_B(i, j) \leq C. \]
\end{problem}

Please note that we choose as OBT the decision version of the problem, thus, if we prove its NP-hardnees the the minimization problem is also NP-hard.

\subsection{Optimization heuristics employed}

As our OBT problem is NP-hard, it will typically be infeasible to aim at an exact solution for realistic input sizes. Instead, it is more reasonable to target approximate solutions by developing various heuristic solutions, possibly with approximation guarantees. These solutions typically include problem-tailored construction procedures in polynomial time, similar to the one proposed in~\cite{7066977} for a similar problem, and meta-heuristic approaches, such as genetic algorithms~\cite{holland} and other randomized optimization heuristics, which can be problem-agnostic or tailored to the problem to various degrees.

For the construction procedures, we are going to slightly extend the approach in~\cite{7066977}, and to propose a more computationally efficient greedy algorithm based on finding a maximum spanning tree. For the meta-heuristic side, we limit ourselves to the simple hill-climber scheme, known also as the $(1+1)$ evolutionary algorithm~\cite{droste-ea}: the best-so-far solution is kept, and the algorithm tries to improve it by introducing small changes, or \emph{mutations}, accepting a mutation if the solution is improved (that is, the cost is decreased). Despite the simplicity, in many cases this scheme is just as efficient as more complicated algorithms. The efficiency is nevertheless limited by the lack of recombination, or crossover, operators, which is a fundamental limitation~\cite{doerr-johannsen-faster-blackbox,jansen-crossover,sudholt-crossover-speeds-up-evco}, but developing a right operator requires additional effort which we leave for the future work.

We also aim at improving the performance of the proposed algorithms at various levels to create a sound baseline. The first of these levels is to avoid, whenever possible, evaluation of the cost function from scratch: approaches like incremental fitness evaluation allow for substantial performance improvements leading to dramatic increase in tractable problem sizes~\cite{bosman-partial-evaluations-ec2021,bosman-linkage-million-variables-gecco17,buzdalovD-gecco17-3cnf,deb-billion-variables} and recently extended even to general crossover operators~\cite{patch-based-ga,pitzerA-logarithmic-fitness-cec2021}. The second level is to design the algorithms tailored to the structure of the problem in the right way, an approach known as gray-box optimization~\cite{next-generation-ga-tutorial}. This approach includes efficient deterministic operators such as the look-ahead mutation~\cite{chicano-moves-in-a-ball-gecco14} and the partition crossover~\cite{whitley-gpx-tsp-ec2020,sanches-whitley-tinos-tsp-px} which are capable of quickly constructing the best solution in a large subspace, leading to huge performance improvements~\cite{whitley-one-million-variables-nk-gecco17,whitley-chicano-goldman}. As we are only in the beginning of exploration of the OBT problem, we limit ourselves with only few of such improvements, however, they already show their usefulness in our experiments.

\section{NP-hardness}
\label{sec:proof}

We present an NP-complete problem which is named Simple Optimal Linear Arrangement Problem (OLA)~\cite{GareyJS76} and then we reduce this problem to our OBT problem and by that we show NP-completeness of our problem.

\begin{problem}[Simple Optimal Linear Arrangement Problem (OLA)]
Let $G = (V, E)$ be an undirected graph with $|V| = n$.
Let $X$ be the required cost of the bijection.
Since $G$ is undirected, we assume that edge $uv$ is the same edge as $vu$.

Question: Does there exist a bijective function $\phi: V \rightarrow \{1, \ldots, n \}$, such that:

\[\sum_{uv \in E} |\phi(u) - \phi(v)| \leq W. \]
\end{problem}

\begin{figure*}[!ht]
  \centering
   \begin{subfigure}[t]{1\textwidth}
    \centering
    \includegraphics[width=0.6\textwidth]{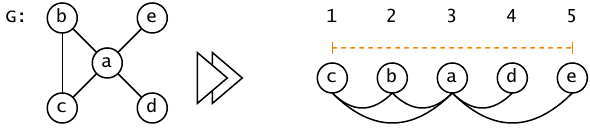}
    \caption{
        A solution of OLA for graph $G$ and labeling cost $W = 7$.
        Orange line represents indexing of vertices.
    }
    \label{fig:example:OLA}
  \end{subfigure}
  
  \hfill\\
  \hfill
  
  \begin{subfigure}[t]{1\textwidth}
    \centering
    \includegraphics[width=0.8\textwidth]{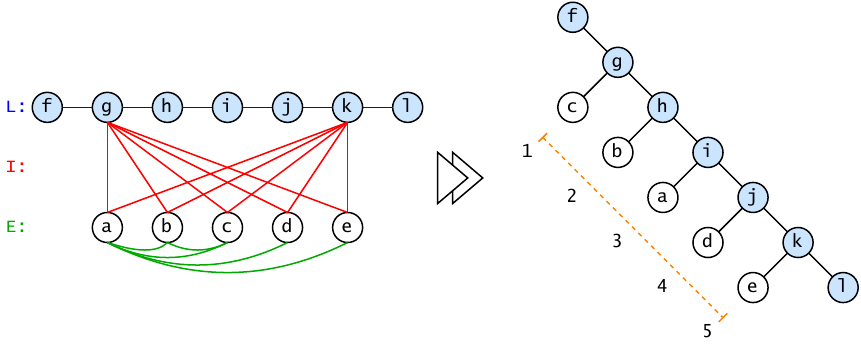}
    \caption{
        A solution of OBT, constructed from OLA instance on (\ref{fig:example:OLA}).
        Vertices $a$ to $e$ represent $V$. Vertices $f$ to $l$ represent $H$ and are highlighted with light blue.
        On the left, edges of $E$ are highlighted with green and have demand $1$ each.
        Edges of $I$ are highlighted with red and have demand $d_1 = 18$.
        Edges of $L$ are highlighted with blue and have demand $d_2 = 558$.
        The right graph is a binary tree with cost $C = 3905$ corresponding to the bijection at Figure~\ref{fig:example:OLA}.
    }
    \label{fig:example:BTA}
  \end{subfigure}
  \caption{%
    The example of OLA instance and an OBT instance reduced from OLA.
  }
  \label{fig:example}
\end{figure*}

\begin{theorem}
OBT problem is NP-complete.
\end{theorem}
\begin{proof}
    It is clear that OBT lies in NP.
    To prove the NP-hardness, we reduce OBT to OLA.

    Consider an instance of OLA: a graph $G = (V, E)$ with $|V| = n$ and $|E| = m$ and the required bijection cost $W$.
    %
    %
    %
    We build an instance of OBT in the following manner (see a left graph on Figure~\ref{fig:example:BTA}).

    %
    We introduce $n + 3$ new vertices $H = \{h_{1}, h_{2}, \ldots, h_{n + 3}\}$.
    Sometimes we refer to a vertex $h_i$ as $v_{n + i}$ and it should be clear from the context.
    %
     
    Then, we introduce two new sets of edges.
    $L = \{(h_i, h_j): h_i, h_j \in H; |i - j| = 1\}$ is a set of edges for a line on vertices from $H$, shown in blue on the left of Figure~\ref{fig:example:BTA}.
    And $I = \{(v_i, h_j): v_i \in V, j \in \{2, n + 2\}\}$ is a set of edges connecting vertices from $V$ to the second and penultimate nodes of $H$, shown in red on the left of Figure~\ref{fig:example:BTA}.

    Then, we introduce a demand matrix $W$ for $2 \cdot n + 3$ vertices in $V \cup H$:

    \begin{equation}\label{eq:A}
        W_{ij}=\begin{cases}
            1, & \text{for $(v_i, v_j) \in E$}.\\
            d_1 = X + 2 \cdot m + 1, & \text{for $(v_i, h_{j - n}) \in I$}.\\
            d_2 = (n^2 + n + 1) \cdot d_1, & \text{for $(h_{i - n}, h_{j - n}) \in L$}.\\
            0, & \text{otherwise}.
        \end{cases}
    \end{equation}

    Finally, we calculate the constant $C$ for our OBT problem.
    \begin{equation}\label{eq:C}
        C = (n + 1) \cdot d_2 + n \cdot (n + 1) \cdot d_1 + X + 2 \cdot m
    \end{equation}

Now we show the correctness of our reduction.

$OLA \Rightarrow OBT$: Suppose there exists a bijection function $\phi$ for $G$ with the cost $\leq X$ then there exists a binary tree with the cost at most $C$.

We construct  binary tree $B$ as shown on the right of Figure~\ref{fig:example:BTA}:
vertices of $H$ form a line and vertices of $V$ are connected to the named line in order $\phi$, starting from the second vertex of $H$.
In other words, vertex $v_i \in V$ is connected to vertex $h_{1 + \phi(i)}$. 
Now to show that this arrangement costs at most $C$:
\begin{itemize}
    \item the sum of costs for edges in $L = (n + 1) \cdot d_2$: each edge traverses distance of $1$~--- the adjacent vertices in $H$ are mapped to the adjacent vertices;
    \item the sum of costs for edges in $I = n \cdot (n + 1) \cdot d_1$: for every $v_i \in V$, $\distg{B}{v_i, h_2} + \distg{B}{v_i, h_{n + 1}} = n + 1$;
    \item the sum of costs for edges in $E \leq W + 2 \cdot m$: each edge $(v_i, v_j)$ traverses distance of $|\phi(i) - \phi(j)| + 2$. 
\end{itemize}
The sum of everything above is at most $C$.

$OBT \Rightarrow OLA$: Suppose there exists a binary tree $B$ for matrix $W$ from Formula~\ref{eq:A} with the cost of at most $C$ from Formula~\ref{eq:C} then there exists OLA with the cost at most $W$.

\begin{enumerate}
    \item Every edge from $L$ traverses the distance of exactly $1$ over $B$. \label{st:1}
    
        We prove this by contradiction: if at least one edge in $L$ traverses a distance of at least $2$, then, all edges of $L$ traverse the distance of at least $(n + 2)$.
        Thus, the total cost of $B$ is at least $(n + 2) \cdot d_2 = (n + 1) \cdot d_2 + (n^2 + n) \cdot d_1 + W + 2 \cdot m + 1 = C + 1$, which contradicts that the cost is at most $C$.

        Thus, vertices in $H$ has to form a line segment in $B$ and in the right order. \label{st:2}

    \item At most one vertex from $V$ can be adjacent to $h_i$ for $2 \leq i \leq n + 1$. \label{st:6}

        Since $B$ is a binary tree, the maximum degree of $h_i$ in $b$ is $3$.
        By Statement~\ref{st:2}, $h_i$ is adjacent to $h_{i - 1}$ and $h_{i + 1}$, which means only one vertex from $V$ can be adjacent to $h_i$ in $B$.

    \item For some $v_i$, the minimal possible value of $x = \distg{B}{v_i, h_2} + \distg{B}{v_i, h_{n + 1}}$ is $n + 1$. \label{st:3}
    
        We have three cases:
        1)~if $v_i$ is adjacent to $h_j$ ($2 \leq j \leq n + 1$), then $x = n + 1$;
        2)~if $v_i$ is adjacent to $h_1$ or $h_{n + 2}$, then $x = n + 3$;
        3)~if $v_i$ is not adjacent to any vertex in $H$, then $x > n + 1$ since $B$ is a tree.

    \item In $B$ every vertex of $V$ is adjacent to one of the vertices from $\{h_2, h_3, \ldots, h_{n + 1}\}$. \label{st:4}

        We prove this statement by contradiction: if some vertex $v_i$ is not adjacent to one of vertices from $\{h_2, h_3, \ldots, h_{n + 1}\}$, then $x = \distg{B}{v_i, h_2} + \distg{B}{v_i, h_{n + 1}} > n + 1$.
        By Statement~\ref{st:3}, all other vertices from $V$ contribute at least $(n + 1) \cdot d_1$ to the cost for edges in $I$.
        It follows that the total cost of edges in $I$ is at least $n \cdot (n + 1) \cdot d_1 + d_1$.
        By Statement~\ref{st:1}, the total cost of edges in $L$ is $(n + 1) \cdot d_2$.
        Hence, the total cost of $B$ is at least $(n + 1) \cdot d_2 + n \cdot (n + 1) \cdot d_1 + W + 2 \cdot m + 1 = C + 1$, which contradicts that the cost is at most $C$.

    \item The arrangement cost of $E$: $\sum_{(v_i, v_j) \in E} \distg{B}{v_i, v_j} \leq W + 2\cdot m$. \label{st:5}

        From Statements~\ref{st:1}~and~\ref{st:4}, the arrangement cost of $L$ and $I$ in $B$ is $(n + 1) \cdot d_2 + n \cdot (n + 1) \cdot d_1$.
        Subtracting that from $C$ we get the upper bound on the total cost of edges in $E$.
\end{enumerate}

Thus, by Statement~\ref{st:4}, every vertex $v_i \in V$ is adjacent to some $h_j$ with $j \in \{2, \ldots, n + 1\}$, and we can define $\phi(i) = j - 1$.
$\phi$ is bijective by Statement~\ref{st:6} and $B$ being a tree.
%
%
%
It follows from the above that
\begin{align*}
\sum\limits_{(v_i, v_j) \in E} |\phi(i) - \phi(j)| &= \sum\limits_{(v_i, v_j) \in E} (\distg{B}{v_i, v_j} - 2) \\&= \sum\limits_{(v_i, v_j) \in E} \distg{B}{v_i, v_j} - 2 \cdot m.    
\end{align*}
By Statement~\ref{st:5}, the cost of $\phi$ does not exceed $W$.
\end{proof}

By that, we proved that OBT is NP-complete.

\section{Optimization heuristics}
\label{sec:heuristics}

In this section we present various optimization heuristics that aim at improving over the existing algorithms, such as a dynamic programming to build the optimal binary search tree topology. Currently our algorithms are designed as a local search that follows the scheme of an $(1+1)$ evolutionary algorithm~\cite{droste-ea}: the best-so-far solution is kept, and the algorithm tries to improve it by introducing small changes, or \emph{mutations}, accepting a mutation if the solution is improved (that is, the cost is decreased). In many cases, we can detect that we have reached a local optimum, so that we can automatically restart the algorithm by sampling a different start point and continuing the search from there while also saving the best found solution.

This way, all our algorithms can be run for any given amount of time, and they only can terminate prematurely if the search space is completely tested. Our experiments involve setting a time limit for each run of each algorithm. A run is forcibly terminated once it reaches the time limit, and the result of an algorithm is considered to be the best solution it could produce before termination.

\subsection{General scheme}

As almost all of our optimization algorithms are designed to run ``forever'', we can represent them internally as infinite streams of pairs $\langle T, F \rangle$, where $T$ is the solution tree and $F$ is its cost computed as~\eqref{eq:cost}.
Instead of a pair, an algorithm can return an empty value \texttt{null}, which can only happen if either the time limit is exceeded or the search space is explored completely.

In this paper, we propose two types of algorithms: the \emph{initializers} which either exhaustively iterates over all combinatorial objects of a certain kind, or samples solutions from a certain distribution, and the \emph{local search} algorithms. An example of an initializer is the algorithm that constructs an optimal binary search tree from~\cite{7066977}, if we run it on a random permutation of vertices. We provide the list of initializers in Section~\ref{sec:algo:initializers}.

The second type of algorithms is a \emph{local search}, which takes an initializer to provide initial solutions, and a procedure that takes a tree and returns another tree, or in other words, a \emph{mutation operator}. We discuss our mutation operators in Section~\ref{sec:algo:mutations}. A mutation operator can also return \texttt{null}, which means that no new trees can be currently produced (for instance, because the entire neighborhood has already been iterated over). Local search can take an advantage of that: if the mutation operator returns \texttt{null}, it can effectively restart by asking the initialization routine for another starting point.

Other than the restart logic, our local search follows the $(1+1)$ scheme common to simple evolutionary algorithms. The current, or \emph{parent} tree is the only tree that survives between iterations. If the new tree generated by the mutation operator is better than the parent tree, the latter is replaced by the former. Since mutation operators may maintain context which is invalidated when the parent tree is changed, the mutation operator is notified of any such change.

Also, we employ techniques to speed up computations, such as the incremental fitness evaluation and the deterministic mutation operators that sample the best outcome. The details on these techniques are given later when discussing particular operators.

\subsection{Cost calculation}\label{sec:cost}

If an algorithm does not calculate the cost implicitly (for example, using an iterative calculation only for a part of the tree), we employ the na{\"i}ve approach if the demand matrix is dense and an approach based on computing lowest common ancestors (LCA)~\cite{lca-definition} if the demand matrix is sparse, thus, reducing the computation time.

The na{\"i}ve approach traverses the tree $n$ times starting with each vertex $v$ where the demand matrix has a non-zero row to compute distances from $v$ in the tree.
Then, it calculates the part of the sum in~\eqref{eq:cost} relevant to $v$ and sum everything up. The running time is $O(n^2)$ and the memory usage is $O(n)$ in this case.

The LCA approach uses a classic reduction to the range minimum query problem. The reduction itself takes $O(n)$ time and space. The next step is a preprocessing that takes $O(n \log n)$ time and space, which allows for subsequent queries that take constant time each~\cite{rmq-loglinspace,rmq-loglinspace2}. Since now LCA for a pair of vertices can be calculated in $O(1)$, the calculation of the cost takes $O(n \log n + m_D)$ where $m_D$ is the number of non-zero entries in the demand matrix. We also implemented a variation of this algorithm that takes $O(n)$ time and space for preprocessing, however, the overall efficiency was worse due to a larger constant factor per query, noting that typically $m_D \gg n$.

Based on preliminary experiments, we propose the following decision on which calculation algorithm to use based on the value of $m_D$: if $m_D \ge n^2 / 2$, we use the na{\"i}ve approach, otherwise, we use the LCA approach. Except for computation of optimal binary search trees and one of the mutation operators, we always use these algorithms to evaluate the cost, as a partial cost evaluation was found to introduce more computational overhead than it could save. Section~\ref{exp:rq:switch} elaborates on the size of the effect.

\subsection{Initializers}\label{sec:algo:initializers}

\subsubsection{Optimal binary search trees, random permutations}\label{algo:bst-rand}
To generate a tree, this algorithm samples a random permutation over $\{1, 2, \ldots, n\}$ and uses the algorithm presented in~\cite{7066977} to construct the ``optimal binary search tree'' using the sampled permutation as the ordering on the vertices. The name ``binary search tree'' refers to the construction strategy, detailed below, which results in construction of a tree that satisfies the search tree property on the chosen order of vertices (that is, all vertices in the left subtree are ``less'' than this vertex, and all vertices in the right subtree are ``greater'', with regards to the chosen order). Now, we describe that algorithm.

Assume we have a permutation of vertex indices $\pi$. For each range $[l; r]$, we first precompute a quantity $O_{l,r}$ which is the sum of the elements $W_{ij}$ of the demand matrix, where $i$ belongs to the set $\{\pi_l, \pi_{l+1}, \ldots, \pi_r\}$ and $j$ does not belong to this set. This can be done in $O(n^3)$ time and $O(n^2)$ space. 

Then, we compute the tree with the smallest cost that can be represented in the following way:
\begin{itemize}
    \item each vertex has an associated integer height $h \ge 0$;
    \item each vertex, except for the unique vertex with the maximum height, has exactly one adjacent vertex (the \emph{parent}) with a larger height;
    \item for every vertex $v$ with a parent $p$, all vertices including $v$ that are reachable from $v$ without visiting $p$ form a contiguous segment in the ordering defined by $\pi$ (that is, they form $\{\pi_l, \pi_{l+1}, \ldots, \pi_r\}$ for some $l \le r$).
\end{itemize}

The constructed tree is a correct binary search tree if each vertex is labelled by its position in the permutation $\pi$, which is the reason to call it the ``optimal binary search tree''. We can construct an optimal tree by deciding for each range $[l; r]$ the optimal cost $C_{l,r}$ and the corresponding best root vertex $R_{l,r}$. This can be done by iterating over $R_{l,r}$, using already computed values of $C$ as well as the values of $O$ mentioned above to compute the cost, and keeping the minimum. The overall running time of this approach is $O(n^3)$ with $O(n^2)$ memory.

Note that this approach would produce the optimal solution, if an appropriate order of vertices is supplied as the permutation. This is, however, sufficiently unlikely for realistic input sizes. Still, even a suboptimal guess can produce a good starting point for subsequent refinements.

\subsubsection{Optimal binary search trees, all permutations}\label{algo:bst-next}
The previous approach can be optimized if the considered permutations are not sampled at random, but rather iterated in a lexicographical order. If we have two permutations $\pi_1$ and $\pi_2$ which are identical from the first position up to position $t$, then in the optimal tree cosntruction we can re-use large portions of both the precomputed array $O$ and the cost/root arrays $C$ and $R$. Indeed, if their second index does not exceed $t$, the existing values would be identical for $\pi_1$ and $\pi_2$, so the old values can be safely kept. This reduces the running time from $O(n^3)$ to $O(n^2 \cdot (n - t))$.

While two randomly sampled permutations $\pi_1$ and $\pi_2$ will have $t=0$ with probability $1 - 1/n$, two lexicographically consecutive permutations have $n - t = O(1)$ on average. Indeed, for any two permutations $\pi_1$ and $\pi_2$, such that $\pi_2$ immediately follows $\pi_1$ in the lexicographical order, only the suffix of the form $v_1 v_2 \ldots v_{n-t}$ is changed, where $v_2 > v_3 > \ldots > v_{n-t}$ and $v_1 < v_2$. The number of consecutive permutation pairs where the length of this suffix is $u := n-t$ is $(u-1) \cdot \binom{n}{u} \cdot u!$, as for the $u$ chosen values in the suffix there are $u-1$ such suffixes due to the choice of $v_1$, there are $\binom{n}{u}$ choices of these values, and the rest of a permutation is arbitrary. Hence the total sum of suffix lengths is
\begin{equation*}
    \sum_{u=2}^{n} u \cdot (u-1) \cdot \binom{n}{u} \cdot u! = \sum_{u=2}^{n} \frac{n!}{(u-2)!} < n! \cdot e,
\end{equation*}
so the average value of $u = n-t$ is smaller than $e = O(1)$.

With this, we can iterate over all permutations in an amortized time $O(n^2)$, not $O(n^3)$, per permutation, so that the whole search can be performed $O(n)$ times faster, in $O(n! \cdot n^2)$ for the whole search space. The improvement in performance is further assessed in Section~\ref{exp:rq:next}.

To reduce the number of used permutations, we can additionally employ various approaches to symmetry breaking. In this work, we use the simplest one: we skip permutations where $\pi_1 > \pi_n$, because the reversed permutation would result in the same optimal tree. More approaches exist: for instance, reversing the part of a permutation that corresponds to a subtree of the optimal tree would result in a construction of the same tree. However, we do not see how to use such heuristics while improving the performance significantly.

\subsubsection{Maximum spanning tree}\label{algo:mst}
The algorithms listed above require time of order $O(n^3)$ to construct any meaningful approximation, which can be too expensive for large networks. Thus, we need alternative approaches which do not guarantee any kind of optimality, but are still good as a starting point.

One such approach is to assign some potential value to each potential edge (an unordered pair of vertices) and then to select edges greedily according to this value, while keeping in mind that we want to obtain a tree with a maximum vertex degree of 3: if the next edge connects already connected parts of the tree, or is adjacent to a vertex with degree 3, it is skipped. Connectivity tests in this case are easily implemented by the Disjoint Set data structure~\cite{Tarjan1984WorstcaseAO}, and the overall algorithm resembles Kruskal's algorithm~\cite{kruskal} for finding minimum spanning trees, subject to additional degree checks.

The particular heuristic of this sort, which we investigate in this paper, considers the demand matrix and treats each entry $W_{ij}$ as the weight of an edge between vertices $i$ and $j$. In order to minimize the sum of $W_{ij} \cdot D_{ij}$, where $D_{ij}$ is the distance between vertices $i$ and $j$ in the resulting tree, we would aim at decreasing $D_{ij}$ for large $W_{ij}$. This naturally results in the algorithm for finding the \emph{maximum} spanning tree, with an additional constraint on the maximum vertex degree.

The running time of the main part of this algorithm is $O(m_D (\log m_D + \alpha(n)))$. We need $O(m_D \log m_D)$ to sort all the potential edges. Then, we add each edge in disjoint sets with the complexity equal to the inverse Ackermann function $\alpha(n)$, leading to $O(m_D \cdot \alpha(n))$ in total. After this part has concluded, we may still have more than one connected component, as the demand graph may have more than one component as well. In this case, we can add arbitrary edges to make a tree, because the particular choice of these edges will not influence the cost of the tree. This part takes at most $O(n)$ time not changing the complexity in the real-life case (that is, $m_D \gg n$).

This approach may be unable to construct an optimal tree, though it can still generate a good starting point. Additionally, in the case where the demand matrix contains multiple equal values, there may be multiple maximum spanning trees. The actual implementation sorts the edges once, but randomly shuffles the edges with equal weights before attempting to construct a tree.

\subsection{Mutation Operators}\label{sec:algo:mutations}

Local searches or evolutionary approaches, like the one employed in this paper, often use mutation operators that introduce local changes and impose a small change of the cost function. We propose three types of algorithms: \emph{edge switch}, \emph{edge replacement}, and \emph{subtree swap}.

\subsubsection{Edge switch}\label{algo:mut-switch}
This operator samples an arbitrary edge $uv$ in the tree. Then, for each vertex $i$ different from $v$ that was adjacent to $u$, the operator removes an edge $ui$ and adds a new edge $vi$, and similarly, for each vertex $j$ different from $u$ that was adjacent to $v$, the operator removes an edge $vj$ and adds a new edge $uj$. This way, the endpoints of the edge appear switched. The example execution of this operator is illustrated in Fig.~\ref{fig:edge-switch}.

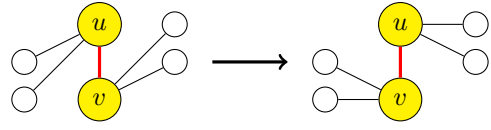
\begin{figure}[!t]
    \centering
    \begin{tikzpicture}[def/.style={draw, circle, minimum width=5pt}]
        \node[def] (Al1) at (0,0.5) {};
        \node[def] (Al2) at (0,1.0) {};
        \node[def,fill=yellow] (Al) at (1,0.5) {$v$};
        \node[def,fill=yellow] (Ar) at (1,1.5) {$u$};
        \node[def] (Ar1) at (2,1.0) {};
        \node[def] (Ar2) at (2,1.5) {};
        \draw (Al1)--(Ar) (Al2)--(Ar) (Al)--(Ar1) (Al)--(Ar2);
        \draw[red, very thick] (Al)--(Ar);

        \draw[->, very thick] (2.5,1)--(3.5,1);

        \node[def] (Bl1) at (4,0.5) {};
        \node[def] (Bl2) at (4,1.0) {};
        \node[def,fill=yellow] (Bl) at (5,0.5) {$v$};
        \node[def,fill=yellow] (Br) at (5,1.5) {$u$};
        \node[def] (Br1) at (6,1.0) {};
        \node[def] (Br2) at (6,1.5) {};
        \draw (Bl1)--(Bl) (Bl2)--(Bl) (Br)--(Br1) (Br)--(Br2);
        \draw[red, very thick] (Bl)--(Br);
    \end{tikzpicture}
    \caption{The edge switch operator}
    \label{fig:edge-switch}
\end{figure}

Consider two connected components that would appear in the tree if the edge $uv$ was to be removed. Denote as $U$ the component that contains vertex $u$, and as $V$ the component that contains vertex $v$. For any pair of vertices $i,j \in U$, the distance between $i$ and $j$ does not change as an effect of this operator, similarly it does not change for any $i,j \in V$. For any $i \in U \setminus \{u\}$ and any $j \in V \setminus \{v\}$, the distance also remains unchanged. The only changes to the cost are due to the following distance changes:
\begin{itemize}
    \item the distance between $i \in U \setminus \{u\}$ and $u$ increases by 1;
    \item the distance between $i \in U \setminus \{u\}$ and $v$ decreases by 1;
    \item the distance between $i \in V \setminus \{v\}$ and $v$ increases by 1;
    \item the distance between $i \in V \setminus \{v\}$ and $u$ decreases by 1.
\end{itemize}

As a result, this operator can be implemented in time $O(n + \mathrm{deg}(W,u) + \mathrm{deg}(W,v))$ including the recalculation of the cost: the time $O(n)$ is needed to mark the connected components $U$ and $V$, and then we need to consider only the edges of the demand graph that are adjacent to $u$ and $v$.

This operator is the only one in the paper which requires time less than $O(m_D)$, so, we expect to be able to explore much more solutions using this operator compared to other operators and algorithms.

Since the number of edges to remove is only $n-1$ for a given tree, we can remember which edges we have previously tried and give up by returning \texttt{null} if all edges have been tried and no improvement has been found. The local search framework will continue the search from a new point sampled by a chosen initialization procedure.

\subsubsection{Edge replacement: random}\label{algo:mut-relinkR}
A mutation which is similarly minor in structure, but more disruptive in the terms of cost changes, is to remove a randomly chosen edge and to connect the two components of the tree with a different edge. The example execution of this operator is illustrated in Fig.~\ref{fig:edge-random-replacement}.

\begin{figure}[!t]
    \centering
    \begin{tikzpicture}[def/.style={draw, circle, minimum width=5pt}]
        \node[def] (Al1) at (0.0,0.0) {};
        \node[def] (Al2) at (0.5,0.0) {};
        \node[def] (Al3) at (0.0,0.5) {};
        \node[def,fill=yellow] (Al4) at (0.5,0.5) {};
        \node[def] (Al5) at (0.5,1.0) {};
        \node[def] (Al6) at (0.0,1.0) {};
        \draw (Al1)--(Al2) (Al2)--(Al3) (Al3)--(Al4) (Al4)--(Al5) (Al5)--(Al6);
        \node[def] (Ar1) at (1.5,0.0) {};
        \node[def] (Ar2) at (2.0,0.0) {};
        \node[def,fill=yellow] (Ar3) at (1.5,0.5) {};
        \node[def] (Ar4) at (2.0,0.5) {};
        \node[def] (Ar5) at (2.0,1.0) {};
        \node[def] (Ar6) at (1.5,1.0) {};
        \draw (Ar1)--(Ar2) (Ar2)--(Ar4) (Ar3)--(Ar4) (Ar4)--(Ar5) (Ar5)--(Ar6);
        \draw[red, very thick] (Al4)--(Ar3);

        \draw[->, very thick] (2.5,0.5)--(3.5,0.5);

        \node[def] (Bl1) at (4.0,0.0) {};
        \node[def] (Bl2) at (4.5,0.0) {};
        \node[def] (Bl3) at (4.0,0.5) {};
        \node[def] (Bl4) at (4.5,0.5) {};
        \node[def,fill=yellow] (Bl5) at (4.5,1.0) {};
        \node[def] (Bl6) at (4.0,1.0) {};
        \draw (Bl1)--(Bl2) (Bl2)--(Bl3) (Bl3)--(Bl4) (Bl4)--(Bl5) (Bl5)--(Bl6);
        \node[def,fill=yellow] (Br1) at (5.5,0.0) {};
        \node[def] (Br2) at (6.0,0.0) {};
        \node[def] (Br3) at (5.5,0.5) {};
        \node[def] (Br4) at (6.0,0.5) {};
        \node[def] (Br5) at (6.0,1.0) {};
        \node[def] (Br6) at (5.5,1.0) {};
        \draw (Br1)--(Br2) (Br2)--(Br4) (Br3)--(Br4) (Br4)--(Br5) (Br5)--(Br6);
        \draw[blue, very thick] (Bl5)--(Br1);
    \end{tikzpicture}
    \caption{The random edge replacement operator}
    \label{fig:edge-random-replacement}
\end{figure}
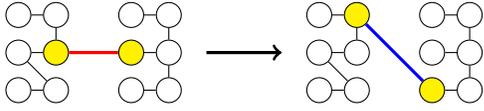

This mutation leaves the distances between vertices within each component unchanged, but can significantly change the distances between vertices in different components. Note that, when sampling the replacement edge, we have to check whether adding this edge violates the degree constraint. However, there is always a constant proportion of the number of edges in each component with degree two or less, so simply resampling invalid edges does not affect the running time.

As we see no easy way to recompute the cost in time less than iterating over all edges of the demand graph, we evaluate the cost using the computation algorithm described in Section~\ref{sec:cost}. As the set of different operator results for a given graph has the size of $O(n^3)$ and is difficult to enumerate, this operator never gives up sampling.

\begin{figure}[!t]
    \centering
    \begin{tikzpicture}[def/.style={draw, circle, minimum width=5pt}]
        \node[def] (Al1) at (0.0,0.0) {};
        \node[def] (Al2) at (0.5,0.0) {};
        \node[def] (Al3) at (0.0,0.5) {};
        \node[def,fill=yellow] (Al4) at (0.5,0.5) {};
        \node[def] (Al5) at (0.5,1.0) {};
        \node[def] (Al6) at (0.0,1.0) {};
        \draw (Al1)--(Al2) (Al2)--(Al3) (Al3)--(Al4);
        \node[def] (Ar1) at (1.5,0.0) {};
        \node[def] (Ar2) at (2.0,0.0) {};
        \node[def,fill=yellow] (Ar3) at (1.5,0.5) {};
        \node[def] (Ar4) at (2.0,0.5) {};
        \node[def] (Ar5) at (2.0,1.0) {};
        \node[def] (Ar6) at (1.5,1.0) {};
        \draw (Ar1)--(Ar2) (Ar2)--(Ar4) (Ar4)--(Ar5) (Ar5)--(Ar6);
        \draw[red, very thick] (Al4)--(Ar3) (Al4)--(Al5) (Al5)--(Al6) (Ar3)--(Ar4);

        \draw[->, very thick] (2.5,0.5)--(3.5,0.5);

        \node[def] (Bl1) at (4.0,0.0) {};
        \node[def] (Bl2) at (4.5,0.0) {};
        \node[def] (Bl3) at (4.0,0.5) {};
        \node[def] (Bl4) at (4.5,0.5) {};
        \node[def,fill=yellow] (Bl5) at (4.5,1.0) {};
        \node[def] (Bl6) at (4.0,1.0) {};
        \draw (Bl1)--(Bl2) (Bl2)--(Bl3) (Bl3)--(Bl4) (Bl4)--(Bl5);
        \node[def] (Br1) at (5.5,0.0) {};
        \node[def] (Br2) at (6.0,0.0) {};
        \node[def,fill=yellow] (Br3) at (5.5,0.5) {};
        \node[def] (Br4) at (6.0,0.5) {};
        \node[def] (Br5) at (6.0,1.0) {};
        \node[def] (Br6) at (5.5,1.0) {};
        \draw (Br1)--(Br2) (Br2)--(Br4) (Br4)--(Br5) (Br5)--(Br6);
        \draw[blue, very thick] (Bl5)--(Br3) (Bl5)--(Bl6) (Br3)--(Br4);
    \end{tikzpicture}\par
    (a) Distance decreases by 1\par if the connection moves towards the vertex \par\vspace{2ex}
    \begin{tikzpicture}[def/.style={draw, circle, minimum width=5pt}]
        \node[def] (Al1) at (0.0,0.0) {};
        \node[def] (Al2) at (0.5,0.0) {};
        \node[def] (Al3) at (0.0,0.5) {};
        \node[def,fill=yellow] (Al4) at (0.5,0.5) {};
        \node[def] (Al5) at (0.5,1.0) {};
        \node[def] (Al6) at (0.0,1.0) {};
        \draw (Al1)--(Al2) (Al2)--(Al3) (Al4)--(Al5) (Al5)--(Al6);
        \node[def] (Ar1) at (1.5,0.0) {};
        \node[def] (Ar2) at (2.0,0.0) {};
        \node[def,fill=yellow] (Ar3) at (1.5,0.5) {};
        \node[def] (Ar4) at (2.0,0.5) {};
        \node[def] (Ar5) at (2.0,1.0) {};
        \node[def] (Ar6) at (1.5,1.0) {};
        \draw (Ar1)--(Ar2) (Ar2)--(Ar4) (Ar4)--(Ar5) (Ar5)--(Ar6);
        \draw[red, very thick] (Al4)--(Ar3) (Ar3)--(Ar4) (Al3)--(Al4);

        \draw[->, very thick] (2.5,0.5)--(3.5,0.5);

        \node[def] (Bl1) at (4.0,0.0) {};
        \node[def] (Bl2) at (4.5,0.0) {};
        \node[def] (Bl3) at (4.0,0.5) {};
        \node[def] (Bl4) at (4.5,0.5) {};
        \node[def,fill=yellow] (Bl5) at (4.5,1.0) {};
        \node[def] (Bl6) at (4.0,1.0) {};
        \draw (Bl1)--(Bl2) (Bl2)--(Bl3) (Bl5)--(Bl6);
        \node[def] (Br1) at (5.5,0.0) {};
        \node[def] (Br2) at (6.0,0.0) {};
        \node[def,fill=yellow] (Br3) at (5.5,0.5) {};
        \node[def] (Br4) at (6.0,0.5) {};
        \node[def] (Br5) at (6.0,1.0) {};
        \node[def] (Br6) at (5.5,1.0) {};
        \draw (Br1)--(Br2) (Br2)--(Br4) (Br4)--(Br5) (Br5)--(Br6);
        \draw[blue, very thick] (Bl5)--(Br3) (Bl3)--(Bl4) (Bl4)--(Bl5) (Br3)--(Br4);
    \end{tikzpicture}\par
    (b) Distance increases by 1\par if the connection moves away from the vertex
    \caption{Optimal edge replacement: the main principle}
    \label{fig:edge-optimal-replacement}
\end{figure}

\subsubsection{Edge replacement: optimal}\label{algo:mut-relinkO}
If we look closer at the random edge replacement operator, we may notice that, although the change of the cost typically involves many edges of the demand graph, the cost change can actually be partitioned in such a way that the choice of the new edge can be performed independently, and quite cheaply, for each of the connected components remaining after the edge removal.

To do this, we consider moving the connection point of the new edge in one of the components to one of the adjacent vertices (see Fig.~\ref{fig:edge-optimal-replacement}). If we split for a moment the vertices of the affected component (the left-hand-side component in Fig.~\ref{fig:edge-optimal-replacement}) into those which have their distance to the connection decreased, and those where this distance is increased, we may note that the cost is decreased by the total demand between the vertices from the first group and the entire second component (the right-hand-side component in the figure), and increased by the total demand between the vertices from the second group and the second component. 

We can easily maintain this demand change while each connection point traverses its respective component by moving along the edges, as in depth-first search, with a relatively easy preprocessing in $O(m_D + n)$ that amounts to computing the total demand in all subtrees assuming an arbitrary vertex to be chosen as a root. As all the changes require only the demand values, but not the particular distances, we can perform component traversals independently of each other. As a result, in $O(m_D + n)$ time we are able to find the edge that connects the two components remaining after removing any given edge, such that the resulting tree has the minimum possible cost.

Similarly to the edge switch operator, and unlike the random edge replacement operator, there are only $n-1$ possible mutation actions, so we can remember which edges have been tested without an improvement and give up if no improvements are possible.

\subsubsection{Subtree swap}\label{algo:mut-subtree}
Finally, we employ an even less local mutation operator, subtree swap, which bears the resemblance with crossover and mutation operators typically employed in genetic programming that often works with representations of programs based on their parse trees~\cite{koza}.

\begin{figure}[!t]
    \centering
    \begin{tikzpicture}[def/.style={draw, circle, minimum width=5pt},
                        del/.style={def, fill=cyan!20!white},
                        der/.style={def, fill=magenta!20!white},
                        dey/.style={def, fill=yellow}]
        \node[del] (A1) at (0.0,0.0) {};
        \node[del] (A2) at (0.0,0.5) {};
        \node[del] (A3) at (0.5,0.0) {};
        \node[dey] (A4) at (0.5,0.5) {};
        \node[dey] (A5) at (1.5,1.5) {};
        \node[der] (A6) at (1.5,2.0) {};
        \node[der] (A7) at (2.0,1.5) {};
        \node[der] (A8) at (2.0,2.0) {};
        \node[def] (AX) at (1.0,1.0) {};
        \node[def] (AY) at (1.3,0.7) {};
        \node[def] (AZ) at (0.7,1.3) {};
        \draw (A1)--(A4) (A2)--(A4) (A3)--(A4) (A5)--(A6) (A5)--(A7) (A5)--(A8);
        \draw (AX)--(AY) (AX)--(AZ);
        \draw[red, very thick] (A4)--(AY) (A5)--(AZ);

        \draw[->, very thick] (2.5,1.0)--(3.5,1.0);

        \node[del] (B1) at (4.0,0.0) {};
        \node[del] (B2) at (4.0,0.5) {};
        \node[del] (B3) at (4.5,0.0) {};
        \node[dey] (B4) at (4.5,0.5) {};
        \node[dey] (B5) at (5.5,1.5) {};
        \node[der] (B6) at (5.5,2.0) {};
        \node[der] (B7) at (6.0,1.5) {};
        \node[der] (B8) at (6.0,2.0) {};
        \node[def] (BX) at (5.0,1.0) {};
        \node[def] (BY) at (5.3,0.7) {};
        \node[def] (BZ) at (4.7,1.3) {};
        \draw (B1)--(B4) (B2)--(B4) (B3)--(B4) (B5)--(B6) (B5)--(B7) (B5)--(B8);
        \draw (BX)--(BY) (BX)--(BZ);
        \draw[blue, very thick] (B4)--(BZ) (B5)--(BY);
    \end{tikzpicture}
    \caption{The subtree swap operator}
    \label{fig:subtree-swap}
\end{figure}
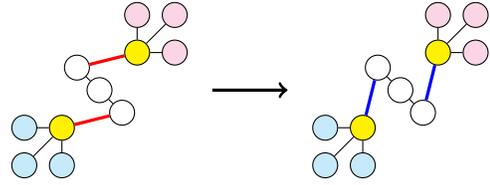

In our case, we sample two different vertices $v_1$ and $v_2$ that are not adjacent to each other, which will be the root vertices of the subtrees to be swapped. Then, we search for the vertices $u_1$ and $u_2$, such that $u_1$ is adjacent to $v_1$, $u_2$ is adjacent to $v_2$, and the path from $v_1$ to $v_2$ in the tree passes through $u_1$ and $u_2$. This can be done straightforwardly with a single depth-first search call in time $O(n)$. After that, we remove edges $v_1u_1$ and $v_2u_2$ and add edges $v_1u_2$ and $v_2u_1$. This is illustrated in Fig.~\ref{fig:subtree-swap}. Due to the complexity of determining the exact changes, we compute the cost from scratch which takes $O(n \log n + m_D)$ time.

For a small number of vertices in the tree, i.e., $n \le 10^3$, we track which pairs of vertices have been tried without an improvement, so that, unless an improvement is found, the operator can sample these pairs without replacement and give up when all pairs have been tried.

\section{Experiments}
\label{sec:exp}
\def\textfraction{0.0} 

\newcommand{\tableFacebook}{
    \hline\multicolumn{5}{c}{Facebook}\\\hline
    MST & 41980 & 42310 & 42503 & 10049208 \\
    BST/rand & 37733 & 37833 & 37900 & 1240767 \\
    BST/next & 38944 & 38944 & 38957 & 15226278 \\\hline
    MST+switch & 41451 & 42332 & 42731 & 610482518 \\
    MST+subtree & 36263 & 36446 & 36504 & 30234427 \\
    MST+replaceR & 36798 & 37188 & 38033 & 29815660 \\
    MST+replaceO & 35955 & 36060 & 36099 & 32082457 \\\hline
    BST+switch & 36681 & 36831 & 36888 & 417062800 \\
    BST+subtree & 35953 & 36059 & 36152 & 30382162 \\
    BST+replaceR & 36982 & 37684 & 38570 & 29323914 \\
    BST+replaceO & 36310 & 36412 & 36498 & 31970291 \\\hline
    MST+random & 34901 & 34965 & 34999 & 39416844 \\
    BST+random & 34935 & 34971 & 35026 & 41628005 \\
}
\newcommand{\tableFacebookBig}{
    \hline\multicolumn{5}{c}{FacebookBig}\\\hline
    MST & 13027298 & 13287066 & 13427598 & 81064 \\
    BST/rand & 9531006 & 9553654 & 9567716 & 1 \\
    BST/next & 9534295 & 9534295 & 9534296 & 446 \\\hline
    MST+switch & 13586432 & 15463896 & 19833956 & 4878540 \\
    MST+subtree & 10039313 & 10511094 & 10985240 & 437597 \\
    MST+replaceR & 13646491 & 13845472 & 14233621 & 426127 \\
    MST+replaceO & 10966297 & 11413588 & 12097809 & 432120 \\\hline
    BST+switch & 9412111 & 9447241 & 9466818 & 87554 \\
    BST+subtree & 9434345 & 9464690 & 9529420 & 93431 \\
    BST+replaceR & 9486279 & 9505413 & 9543037 & 84131 \\
    BST+replaceO & 9377749 & 9407370 & 9429707 & 100154 \\\hline
    MST+random & 9324367 & 9403820 & 9561352 & 590351 \\
    BST+random & 9231446 & 9294316 & 9316414 & 180777 \\
}
\newcommand{\tableHPC}{
    \hline\multicolumn{5}{c}{HPC}\\\hline
    MST & 2821329 & 2873195 & 2898571 & 7952388 \\
    BST/rand & 5460803 & 5493353 & 5509732 & 9607 \\
    BST/next & 5199117 & 5199117 & 5199117 & 695648 \\\hline
    MST+switch & 2905612 & 2957631 & 2999053 & 152493590 \\
    MST+subtree & 2126204 & 2150038 & 2173537 & 22634544 \\
    MST+replaceR & 2307009 & 2366124 & 2444087 & 21412024 \\
    MST+replaceO & 2206862 & 2217114 & 2226677 & 21076405 \\\hline
    BST+switch & 5245504 & 5306637 & 5320441 & 41993097 \\
    BST+subtree & 2789394 & 2894296 & 3040571 & 21854513 \\
    BST+replaceR & 3461280 & 3629122 & 3765951 & 21086333 \\
    BST+replaceO & 3203043 & 3279672 & 3321721 & 16748357 \\\hline
    MST+random & 1994996 & 2010771 & 2023985 & 27459469 \\
    BST+random & 2090163 & 2166391 & 2279279 & 27330167 \\
}
\newcommand{\tableMicrosoft}{
    \hline\multicolumn{5}{c}{Microsoft}\\\hline
    MST & 195765 & 196065 & 196280 & 19934065 \\
    BST/rand & 190126 & 190666 & 191350 & 1225785 \\
    BST/next & 215376 & 215376 & 215376 & 16410909 \\\hline
    MST+switch & 189602 & 190164 & 190583 & 633031106 \\
    MST+subtree & 176187 & 176652 & 177033 & 54096341 \\
    MST+replaceR & 181658 & 183873 & 185530 & 53918530 \\
    MST+replaceO & 177016 & 177387 & 177680 & 56465216 \\\hline
    BST+switch & 187495 & 188852 & 189294 & 277441413 \\
    BST+subtree & 176058 & 176771 & 177238 & 53817763 \\
    BST+replaceR & 186121 & 197414 & 206078 & 52487717 \\
    BST+replaceO & 179987 & 181303 & 181924 & 51765806 \\\hline
    MST+random & 174184 & 174312 & 174493 & 69428875 \\
    BST+random & 174243 & 174604 & 174862 & 72902931 \\
}
\newcommand{\tableProjecToR}{
    \hline\multicolumn{5}{c}{ProjecToR}\\\hline
    MST & 2144182 & 2151455 & 2154584 & 13124929 \\
    BST/rand & 2033936 & 2042839 & 2046711 & 744315 \\
    BST/next & 2284848 & 2284848 & 2284848 & 11690142 \\\hline
    MST+switch & 2079812 & 2088527 & 2092005 & 578565562 \\
    MST+subtree & 1879695 & 1885425 & 1888427 & 36256546 \\
    MST+replaceR & 1927741 & 1953622 & 1967282 & 36458395 \\
    MST+replaceO & 1893249 & 1895969 & 1898579 & 38214681 \\\hline
    BST+switch & 2009359 & 2018686 & 2024121 & 233016137 \\
    BST+subtree & 1880980 & 1892171 & 1898173 & 35240270 \\
    BST+replaceR & 2010122 & 2064176 & 2148628 & 34502627 \\
    BST+replaceO & 1928045 & 1951338 & 1959583 & 35566396 \\\hline
    MST+random & 1861339 & 1862483 & 1864095 & 47033703 \\
    BST+random & 1863430 & 1866897 & 1869250 & 47210798 \\
}
\newcommand{\tableartA}{
    \hline\multicolumn{5}{c}{$\alpha=0.0$}\\\hline
    MST & 314966 & 323023 & 327927 & 6181557 \\
    BST/rand & 1223682 & 1230622 & 1234015 & 1256 \\
    BST/next & 1251028 & 1251028 & 1251028 & 175206 \\\hline
    MST+switch & 332057 & 345866 & 352701 & 76768691 \\
    MST+subtree & 182071 & 190354 & 204405 & 17319181 \\
    MST+replaceR & 119237 & 123601 & 131398 & 14449684 \\
    MST+replaceO & 114134 & 115407 & 116776 & 12329723 \\\hline
    BST+switch & 1074886 & 1099451 & 1112937 & 32285913 \\
    BST+subtree & 292192 & 321449 & 346494 & 16802917 \\
    BST+replaceR & 291655 & 311257 & 341984 & 14247071 \\
    BST+replaceO & 250019 & 260324 & 267751 & 8269413 \\\hline
    MST+random & 111259 & 113083 & 120261 & 18177272 \\
    BST+random & 111220 & 116011 & 126771 & 17619055 \\
}
\newcommand{\tableartB}{
    \hline\multicolumn{5}{c}{$\alpha=0.25$}\\\hline
    MST & 320548 & 328397 & 333494 & 6430813 \\
    BST/rand & 1221781 & 1229798 & 1232179 & 1252 \\
    BST/next & 1249356 & 1249356 & 1249356 & 182411 \\\hline
    MST+switch & 319099 & 354021 & 363536 & 73980705 \\
    MST+subtree & 176365 & 194014 & 202627 & 17454896 \\
    MST+replaceR & 120966 & 125846 & 132771 & 14561093 \\
    MST+replaceO & 112918 & 114974 & 116345 & 12398337 \\\hline
    BST+switch & 1082557 & 1106856 & 1112944 & 30046662 \\
    BST+subtree & 296425 & 314338 & 332775 & 16925828 \\
    BST+replaceR & 287111 & 309053 & 329474 & 14329092 \\
    BST+replaceO & 238495 & 256396 & 264889 & 8437018 \\\hline
    MST+random & 111181 & 114013 & 122363 & 17481594 \\
    BST+random & 111189 & 113964 & 133407 & 17628255 \\
}
\newcommand{\tableartC}{
    \hline\multicolumn{5}{c}{$\alpha=0.5$}\\\hline
    MST & 305864 & 315296 & 319634 & 6406305 \\
    BST/rand & 1223361 & 1228354 & 1231660 & 1284 \\
    BST/next & 1246782 & 1246782 & 1246782 & 177950 \\\hline
    MST+switch & 320197 & 340293 & 345312 & 77539430 \\
    MST+subtree & 178388 & 186650 & 197361 & 17596166 \\
    MST+replaceR & 115912 & 121764 & 126822 & 14826849 \\
    MST+replaceO & 112641 & 113772 & 115297 & 12494553 \\\hline
    BST+switch & 1077832 & 1097265 & 1106260 & 34299502 \\
    BST+subtree & 283217 & 307189 & 332545 & 17054650 \\
    BST+replaceR & 277118 & 298349 & 328668 & 14354453 \\
    BST+replaceO & 240446 & 253749 & 260189 & 8459932 \\\hline
    MST+random & 110589 & 112214 & 114976 & 19435618 \\
    BST+random & 110844 & 113319 & 124428 & 19146924 \\
}
\newcommand{\tableartD}{
    \hline\multicolumn{5}{c}{$\alpha=0.75$}\\\hline
    MST & 292260 & 302487 & 309321 & 6206266 \\
    BST/rand & 1216492 & 1220815 & 1226334 & 1234 \\
    BST/next & 1242278 & 1242278 & 1242278 & 174685 \\\hline
    MST+switch & 318988 & 326270 & 334701 & 75145051 \\
    MST+subtree & 171198 & 179327 & 195025 & 17426273 \\
    MST+replaceR & 113440 & 119278 & 124543 & 14410834 \\
    MST+replaceO & 110769 & 112077 & 112548 & 12353661 \\\hline
    BST+switch & 1070368 & 1080898 & 1088404 & 37177750 \\
    BST+subtree & 268020 & 294621 & 323000 & 16877155 \\
    BST+replaceR & 278113 & 304875 & 327032 & 14252785 \\
    BST+replaceO & 220099 & 247229 & 256977 & 8322455 \\\hline
    MST+random & 109750 & 111528 & 117198 & 17247531 \\
    BST+random & 109703 & 111683 & 115701 & 17255646 \\
}
\newcommand{\tableartE}{
    \hline\multicolumn{5}{c}{$\alpha=0.9$}\\\hline
    MST & 258063 & 262088 & 266126 & 6480083 \\
    BST/rand & 1199839 & 1208647 & 1211149 & 1288 \\
    BST/next & 1227895 & 1227895 & 1227895 & 178337 \\\hline
    MST+switch & 274159 & 282532 & 289180 & 78279894 \\
    MST+subtree & 149626 & 163165 & 172017 & 17240927 \\
    MST+replaceR & 111212 & 114468 & 118496 & 14629868 \\
    MST+replaceO & 108618 & 109019 & 109519 & 12252616 \\\hline
    BST+switch & 1053371 & 1066404 & 1077171 & 32739728 \\
    BST+subtree & 225512 & 250357 & 274893 & 16688458 \\
    BST+replaceR & 256368 & 274710 & 293642 & 14305024 \\
    BST+replaceO & 204653 & 214146 & 226227 & 8361872 \\\hline
    MST+random & 108127 & 108833 & 111724 & 19314674 \\
    BST+random & 108117 & 109084 & 113176 & 18646159 \\
}
\newcommand{\tablepFabric}{
    \hline\multicolumn{5}{c}{pFabric}\\\hline
    MST & 385603 & 396925 & 404627 & 6664364 \\
    BST/rand & 352828 & 353091 & 353250 & 1268820 \\
    BST/next & 355385 & 355385 & 355385 & 15159475 \\\hline
    MST+switch & 389540 & 401035 & 406146 & 585701743 \\
    MST+subtree & 352406 & 353847 & 354540 & 21233235 \\
    MST+replaceR & 354150 & 355932 & 357463 & 21706670 \\
    MST+replaceO & 351134 & 351758 & 352093 & 23418446 \\\hline
    BST+switch & 348943 & 349509 & 349728 & 460648062 \\
    BST+subtree & 347561 & 348282 & 348496 & 21112944 \\
    BST+replaceR & 353463 & 354961 & 356201 & 20470019 \\
    BST+replaceO & 352049 & 352384 & 352595 & 23609006 \\\hline
    MST+random & 344974 & 345269 & 345366 & 27934373 \\
    BST+random & 344780 & 345305 & 345418 & 27793645 \\
}

Based on the heuristics outlined in Section~\ref{sec:heuristics}, we test the following algorithms:
\begin{itemize}
    \item ``MST'': repeated calls to the maximum spanning tree heuristic, as in Section~\ref{algo:mst}, with edges of equal weights randomly shuffled between the calls.
    \item ``BST/rand'': the construction of optimal binary search trees, as in Section~\ref{algo:bst-rand}, using randomly generated permutations of vertices.
    \item ``BST/next'': the construction of optimal binary search trees using lexicographical enumeration of permutations of vertices, as in Section~\ref{algo:bst-next}.
    \item ``MST+switch'': the local search using the MST heuristic as the initialization procedure and the edge switch mutation, as in Section~\ref{algo:mut-switch}.
    \item ``MST+subtree'': the local search using the MST heuristic as the initialization procedure and the subtree swap mutation, as in Section~\ref{algo:mut-subtree}.
    \item ``MST+replaceR'': the local search using the MST heuristic as the initialization procedure and the random edge replacement mutation, as in Section~\ref{algo:mut-relinkR}.
    \item ``MST+replaceO'': the local search using the MST heuristic as the initialization procedure and the optimal edge replacement mutation, as in Section~\ref{algo:mut-relinkO}.
    \item ``MST+random'': the local search using the MST heuristic as the initialization procedure and one of the following three mutation operators chosen at random each time: the edge switch, the subtree swap, and the optimal edge replacement.
    \item ``BST+switch'', ``BST+subtree'', ``BST+replaceR'', ``BST+replaceO'', ``BST+random'': same as the above but using the BST/rand heuristic as the initialization procedure.
\end{itemize}

To investigate the performance of these algorithms, we use two groups of tests, where each test is essentially a demand graph: the synthetic tests and the real-world tests.

The synthetic tests are adapted from~\cite{avin2020complexity} where they test the properties of dynamically adjusting demand-aware networks. They are characterized by the set of possible vertex pairs between which a message is assumed to be passed, and the temporal locality parameter $\alpha$: the next request pair is chosen with probability $\alpha$ while a random pair is chosen otherwise. In the context of this work, the temporal behavior is not considered, so the demand matrix element $W_{ij}$ is equal to the number of messages passed between $i$ and $j$ in the test. We use five such tests with the same number of vertices $1023$ and the temporal locality parameters $\alpha \in \{ 0.0, 0.25, 0.5, 0.75, 0.9\}$.

The following six real-world tests, adapted from their respective sources similarly to the above, are used:
\begin{itemize}
    \item The small and the large tests from~\cite{roy2015inside}, which we refer to further as, correspondingly, ``Facebook'' (100 vertices, 2990 edges) and ``FacebookBig'' (10000 vertices, 151677 edges).
    \item Two tests from~\cite{ghobadi2016projector} referred to as ``ProjecToR'' (121 vertices, 2322 edges) and ``Microsoft'' (100 vertices, 1431 edges).
    \item The test from~\cite{doe2016characterization} referred to as ``HPC'' (544 vertices, 1620 edges).
    \item The test from~\cite{alizadeh2013pfabric} referred to as ``pFabric'' (100 vertices, 4506 edges).
\end{itemize}

For each algorithm and each test we conducted $20$ independent runs with a time limit of one hour ($3600$ seconds). Note that all mentioned algorithms are capable of running for indefinite amount of time, except for ``BST/next'' which may terminate after testing all permutations of vertices, but none of the tests were sufficiently small for this to happen.

To conduct these runs, we used two identical computers, each with two AMD Opteron (TM) 6272 CPUs amounting to 32 cores clocked at 2.1 GHz, and with 256 gigabytes of memory available. All algorithms were implemented in Java and ran using the OpenJDK VM version 17.0.9. The anonymized version of the code is available in a dedicated repository\footnote{\url{https://anonymous.4open.science/r/tree-for-network-1973}}.

\begin{table}[!t]
\caption{Costs for synthetic tests}
\label{tbl:cost:synth}
\centering
\scalebox{0.97}{
\begin{tabular}{l|rrr|r}
Algorithm & Min & Median & Max & Median queries\\
\tableartA
\tableartB
\tableartC
\tableartD
\tableartE
\end{tabular}
}
\end{table}

\begin{table}[!t]
\caption{Costs for smaller real-world tests}
\label{tbl:cost:real}
\centering
\scalebox{0.97}{
\begin{tabular}{l|rrr|r}
Algorithm & Min & Median & Max & Median queries \\
\tableFacebook
\tableMicrosoft
\tableHPC
\tableProjecToR
\tablepFabric
\end{tabular}
}
\end{table}

\begin{table}[!t]
\caption{Costs for the biggest real-world test}
\label{tbl:cost:real-big}
\centering
\begin{tabular}{l|rrr|r}
Algorithm & Min & Median & Max & Median queries \\
\tableFacebookBig
\end{tabular}
\end{table}

The results are presented in Table~\ref{tbl:cost:synth} for synthetic tests, Table~\ref{tbl:cost:real-big} for the FacebookBig test, and Table~\ref{tbl:cost:real} for all other real-world tests. We show minimum, median and maximum costs for each combination of a test and an algorithm, as well as the median number of generated solutions in the column named ``Median queries''. The median for both costs and queries is chosen instead of mean, because the distributions of randomized search heuristic outcomes are typically far from being normal~\cite{stattest}, so for all subsequent statistical tests we use a non-parametric test (namely, the Wilcoxon rank sum test~\cite{wilcoxon-tests,mann-whitney-u-test}, for which we use the implementation available in R~\cite{R}).

As a final comparison, we consider ``MST+random'' and ``BST+random'', as our best evolutionary algorithms, against ``MST'' and ``BST/rand''. For synthetic datasets, the improvement is almost three fold. For the small real datasets the results are a little bit controversial: the smallest improvement is just 3\% on pFabric which nevertheless can give a large boost for the datacenter, while, for example, for HPC the improvement is almost 1.5x. The average improvement is approximately $10\%$ which we consider significant. For our largest dataset, ``FacebookBig'', the improvement is small (approximately, $1\%$)~--- we explain this underperformance due to the fact that our algorithms made very small amount of iterations in one hour on such large graph, e.g., only one BST construction on a random permutation fits into this time-limit.

Next, we investigate a few research statistical questions about our algorithms using these data whether our local search makes sense.

\subsection{Does the local search improve the results?}

We ran Wilcoxon rank sum test on final costs of ``MST'' and all ``MST+'' algorithms, as well as of ``BST/rand'' and all ``BST+'' algorithms.

For the synthetic tests, ``MST+switch'' was clearly worse than ``MST'', and all other local search algorithms were clearly better than ``MST'', with $p$-value of $7.25 \cdot 10^{-12}$ in most of the cases, indicating that the ranges of samples did not intersect. The only exception was the case of ``MST+switch'' with $\alpha=0.25$ ($p = 9.984 \cdot 10^{-7}$, and of ``MST+replaceO'' with $\alpha=0.9$ ($p = 3.393 \cdot 10^{-8}$), where in the latter case it was solely due to multiple coinciding values.
All ``BST+'' versions were clearly better ($p = 7.25 \cdot 10^{-12}$) than ``BST/rand''.

For real-world cases, the results were similar at the level of $p = 3.8 \cdot 10^{-8}$, but, additionally, occasionally random edge replacement did not reach significance or even became worse. This happened with ``MST+replaceR'' on ``FacebookBig'' and with ``BST+replaceR'' on ``Microsoft'', ``ProjecToR'' and ``pFabric'' at various significance levels less than $0.043$.

As a result, we can see that there are mutation operators, such as subtree swap or optimal edge replacement, which, together with the local search framework, provide better performance than just using the initialization procedure indefinitely.

\subsection{Is edge switch faster than other mutations?}
\label{exp:rq:switch}

In Section~\ref{sec:cost}, we mentioned that the edge switch mutation uses a custom cost evaluation procedure, while other mutations compute the cost from scratch. Here, we can see that, in most cases, edge switch is indeed faster than other mutations at the significance level of $p=7.25 \cdot 10^{-12}$, and, according to the last column of the tables, the median number of queries is 3--20 times bigger, which indicates higher computational efficiency of edge switch.

The only exceptions are the random mutations which are still slower but at times slightly less significantly ($p < 2.5 \cdot 10^{-6}$), and the ``FacebookBig'' test together with BST-based algorithms. The latter can be explained by the fact that the time limit of 1~hour on such large input allowed only for one run of the BST initialization, and the remaining time was highly varying which hindered all the differences. 

Overall the increased computational efficiency of edge switch (as opposed to its effectiveness) is confirmed.

\subsection{BST: random or lexicographical enumeration?}
\label{exp:rq:next}

The number of queries of ``BST/next'' is always bigger than of ``BST/rand'' at the significance level $p < 3.393 \cdot 10^{-8}$, and typically the difference is huge (446 vs 1 query for ``FacebookBig'', almost 15 times for synthetic tests). However, the cost in the case of ``BST/next'' is better only in two cases (``HPC'' and ``FacebookBig'', $p<4.1 \cdot 10^{-9}$), and worse in other cases ($p<9.8 \cdot 10^{-9}$).

\subsection{Edge replacement: Is optimal better than random?}

Comparing ``BST+replaceR'' and ``BST+replaceO'', as well as ``MST+replaceR'' and ``MST+replaceO'', showed that the optimal edge replacement is always better at the significance level $p < 3.4 \cdot 10^{-8}$.

\subsection{Which mutation is the best?}

If not taking the random mutation choice into account, within BST-based and MST-based local search algorithms, the subtree swap mutation and optimal edge replacement are the main competitors. For the case of synthetic tests, the optimal edge replacement is the clear winner, whereas in other cases, the optimal edge replacement was better on three workload and worse on others. All these comparisons are significant at level $p < 3.4 \cdot 10^{-8}$.

However, the random choice out of three mutation operators is the clear winner in all the cases at the same significance level. This indicates that different mutation operators augment each other, most likely by employing different neighborhood metrics, such that if a solution is locally optimal (or simply hard to escape) for one of the mutation operators, another one can find the improving move faster.

\subsection{Which initialization is the best?}

Taking into account the previous section, it only makes sense to compare ``BST+random'' and ``MST+random''. For the synthetic results, they are mostly statistically indistinguishable (with an exception of $\alpha=0.5$, where MST is better at $p=0.0012$). For the real-world data, MST is better for ``HPC'', ``Microsoft'' and ``ProjecToR'' ($p < 3 \cdot 10^{-6}$), while BST is better for ``FacebookBig'' ($p = 3.328 \cdot 10^{-8}$). It seems that in practice it makes sense to consider both of them.

\section{Conclusion}
\label{sec:conclusion}

To summarize, we showed that the construction of an optimal demand-aware binary tree network is NP-hard. Then, we presented new effective randomized algorithms that search through the space of possible answers. Finally, we used an $(1+1)$ evolutionary algorithm with newly designed mutations which give an improvement of $10\%$ on average and up to $30\%$ on real-world workloads over our randomized algorithms.

As for future work, it would be interesting to investigate other mutations, such as replacing several edges at the same time in an optimal manner, and investigate the possibility to make crossovers, possibly also deterministic and optimal, between two binary tree networks. Some of our results, such as the unguided random subtree swap mutation being on par or better than the optimal edge replacement mutation, suggest that operators that are capable of big restructuring moves can be beneficial. Finally, we would like to generalize our algorithms for $k$-ary trees and more complicated networks with bounded degrees.

\bibliographystyle{abbrv}
\bibliography{references}

\begin{thebibliography}{10}

\bibitem{lca-definition}
A.~Aho, J.~Hopcroft, and J.~Ullman.
\newblock On finding lowest common ancestors in trees.
\newblock In {\em Proceedings of 5th ACM Symposium on Theory of Computing},
  pages 253--265, 1973.

\bibitem{alizadeh2013pfabric}
M.~Alizadeh, S.~Yang, M.~Sharif, S.~Katti, N.~McKeown, B.~Prabhakar, and
  S.~Shenker.
\newblock pfabric: Minimal near-optimal datacenter transport.
\newblock {\em ACM SIGCOMM Computer Communication Review}, 43(4):435--446,
  2013.

\bibitem{avin2020complexity}
C.~Avin, M.~Ghobadi, C.~Griner, and S.~Schmid.
\newblock On the complexity of traffic traces and implications.
\newblock {\em Proceedings of the ACM on Measurement and Analysis of Computing
  Systems}, 4(1):1--29, 2020.

\bibitem{avin2018rdan}
C.~Avin, A.~Hercules, A.~Loukas, and S.~Schmid.
\newblock {rDAN}: Toward robust demand-aware network designs.
\newblock {\em Information Processing Letters}, 133:5--9, 2018.

\bibitem{avin2020demand}
C.~Avin, K.~Mondal, and S.~Schmid.
\newblock Demand-aware network designs of bounded degree.
\newblock {\em Distributed Computing}, 33(3-4):311--325, 2020.

\bibitem{avin2022demand}
C.~Avin, K.~Mondal, and S.~Schmid.
\newblock Demand-aware network design with minimal congestion and route
  lengths.
\newblock {\em IEEE/ACM Transactions on Networking}, 30(4):1838--1848, 2022.

\bibitem{avin2018demandaware}
C.~Avin and S.~Schmid.
\newblock Toward demand-aware networking: A theory for self-adjusting networks,
  2018.

\bibitem{rmq-loglinspace}
M.~Bender and M.~Farach-Colton.
\newblock The {LCA} problem revisited.
\newblock In {\em LATIN 2000: Theoretical Informatics}, number 1776 in Lecture
  Notes in Computer Science, pages 88--94. Springer, 2000.

\bibitem{rmq-loglinspace2}
M.~A. Bender, M.~Farach-Colton, G.~Pemmasani, S.~Skiena, and P.~Sumazin.
\newblock Lowest common ancestors in trees and directed acyclic graphs.
\newblock {\em Journal of Algorithms}, 57(2):75--94, 2005.

\bibitem{bosman-partial-evaluations-ec2021}
A.~Bouter, T.~Alderliesten, and P.~A. Bosman.
\newblock Achieving highly scalable evolutionary real-valued optimization by
  exploiting partial evaluations.
\newblock {\em Evolutionary Computation}, 29(1):129--155, 2021.

\bibitem{bosman-linkage-million-variables-gecco17}
A.~Bouter, T.~Alderliesten, C.~Witteveen, and P.~A.~N. Bosman.
\newblock Exploiting linkage information in real-valued optimization with the
  real-valued gene-pool optimal mixing evolutionary algorithm.
\newblock In {\em Proceedings of Genetic and Evolutionary Computation
  Conference}, pages 705--712, 2017.

\bibitem{patch-based-ga}
M.~Buzdalov.
\newblock Improving time and memory efficiency of genetic algorithms by storing
  populations as minimum spanning trees of patches.
\newblock In {\em Proceedings of Genetic and Evolutionary Computation
  Conference Companion}, pages 1873--1881, 2023.

\bibitem{buzdalovD-gecco17-3cnf}
M.~Buzdalov and B.~Doerr.
\newblock Runtime analysis of the {$(1+(\lambda,\lambda))$} genetic algorithm
  on random satisfiable {3-CNF} formulas.
\newblock In {\em Proceedings of Genetic and Evolutionary Computation
  Conference}, pages 1343--1350, 2017.

\bibitem{whitley-one-million-variables-nk-gecco17}
F.~Chicano, D.~Whitley, G.~Ochoa, and R.~Tin{\'o}s.
\newblock Optimizing one million variable {NK} landscapes by hybridizing
  deterministic recombination and local search.
\newblock In {\em Proceedings of Genetic and Evolutionary Computation
  Conference}, pages 753--760, 2017.

\bibitem{chicano-moves-in-a-ball-gecco14}
F.~Chicano, D.~Whitley, and A.~M. Sutton.
\newblock Efficient identification of improving moves in a ball for
  pseudo-{B}oolean problems.
\newblock In {\em Proceedings of Genetic and Evolutionary Computation
  Conference}, pages 437--444, 2014.

\bibitem{deb-billion-variables}
K.~Deb and C.~Myburgh.
\newblock A population-based fast algorithm for a billion-dimensional resource
  allocation problem with integer variables.
\newblock {\em European Journal of Operational Research}, 261(2):460--474,
  2017.

\bibitem{stattest}
J.~Derrac, S.~Garcia, D.~Molina, and F.~Herrera.
\newblock A practical tutorial on the use of nonparametric statistical tests as
  a methodology for comparing evolutionary and swarm intelligence algorithms.
\newblock {\em Swarm and Evolutionary Computation}, 1(1):3--18, 2011.

\bibitem{doe2016characterization}
U.~DOE.
\newblock Characterization of the {DOE} mini-apps.
\newblock {\em Retrieved July}, 14, 2016.

\bibitem{doerr-johannsen-faster-blackbox}
B.~Doerr, D.~Johannsen, T.~K\"otzing, P.~K. Lehre, M.~Wagner, and C.~Winzen.
\newblock Faster black-box algorithms through higher arity operators.
\newblock In {\em Proceedings of Foundations of Genetic Algorithms}, pages
  163--172, 2011.

\bibitem{droste-ea}
S.~Droste, T.~Jansen, and I.~Wegener.
\newblock On the analysis of the {(1+1)} evolutionary algorithm.
\newblock {\em Theoretical Computer Science}, 276(1-2):51--81, 2002.

\bibitem{farrington2010helios}
N.~Farrington, G.~Porter, S.~Radhakrishnan, H.~H. Bazzaz, V.~Subramanya,
  Y.~Fainman, G.~Papen, and A.~Vahdat.
\newblock Helios: a hybrid electrical/optical switch architecture for modular
  data centers.
\newblock In {\em Proceedings of the ACM SIGCOMM 2010 Conference}, pages
  339--350, 2010.

\bibitem{foerster2018characterizing}
K.-T. Foerster, M.~Ghobadi, and S.~Schmid.
\newblock Characterizing the algorithmic complexity of reconfigurable data
  center architectures.
\newblock In {\em Proceedings of the 2018 Symposium on Architectures for
  Networking and Communications Systems}, pages 89--96, 2018.

\bibitem{GareyJS76}
M.~R. Garey, D.~S. Johnson, and L.~J. Stockmeyer.
\newblock Some simplified {NP}-complete graph problems.
\newblock {\em Theor. Comput. Sci.}, 1(3):237--267, 1976.

\bibitem{ghobadi2016projector}
M.~Ghobadi, R.~Mahajan, A.~Phanishayee, N.~Devanur, J.~Kulkarni, G.~Ranade,
  P.-A. Blanche, H.~Rastegarfar, M.~Glick, and D.~Kilper.
\newblock {ProjecToR}: Agile reconfigurable data center interconnect.
\newblock In {\em Proceedings of the 2016 ACM SIGCOMM Conference}, pages
  216--229, 2016.

\bibitem{hamedazimi2014firefly}
N.~Hamedazimi, Z.~Qazi, H.~Gupta, V.~Sekar, S.~R. Das, J.~P. Longtin, H.~Shah,
  and A.~Tanwer.
\newblock Firefly: A reconfigurable wireless data center fabric using
  free-space optics.
\newblock In {\em Proceedings of the 2014 ACM conference on SIGCOMM}, pages
  319--330, 2014.

\bibitem{holland}
J.~H. Holland.
\newblock {\em Adaptation in Natural and Artificial Systems}.
\newblock University of Michigan, 1975.

\bibitem{jansen-crossover}
T.~Jansen and I.~Wegener.
\newblock The analysis of evolutionary algorithms---a~proof that crossover
  really can help.
\newblock {\em Algorithmica}, 34:47--66, 2002.

\bibitem{koza}
J.~R. Koza.
\newblock {\em Genetic programming: on the programming of computers by means of
  natural selection}.
\newblock MIT Press, Cambridge, MA, USA, 1992.

\bibitem{kruskal}
J.~B. Kruskal.
\newblock On the shortest spanning subtree of a graph and the traveling
  salesman problem.
\newblock {\em Proceedings of the American Mathematical Society}, 7(1):48--50,
  1956.

\bibitem{leiserson1985fat}
C.~E. Leiserson.
\newblock Fat-trees: Universal networks for hardware-efficient supercomputing.
\newblock {\em IEEE Transactions on Computers}, 100(10):892--901, 1985.

\bibitem{liu2014circuit}
H.~Liu, F.~Lu, A.~Forencich, R.~Kapoor, M.~Tewari, G.~M. Voelker, G.~Papen,
  A.~C. Snoeren, and G.~Porter.
\newblock Circuit switching under the radar with {REACToR}.
\newblock In {\em 11th USENIX Symposium on Networked Systems Design and
  Implementation (NSDI 14)}, pages 1--15, 2014.

\bibitem{mann-whitney-u-test}
H.~B. Mann and D.~R. Whitney.
\newblock On a test of whether one of two random variables is stochastically
  larger than the other.
\newblock {\em Annals of Mathematical Statistics}, 18(1):50--60, 1947.

\bibitem{pitzerA-logarithmic-fitness-cec2021}
E.~Pitzer and M.~Affenzeller.
\newblock Cheating like the neighbors: Logarithmic complexity for fitness
  evaluation in genetic algorithms.
\newblock In {\em Proceedings of IEEE Congress on Evolutionary Computation},
  pages 1431--1438, 2021.

\bibitem{R}
{R Core Team.}
\newblock R: A language and environment for statistical computing.
\newblock http://www.R-project.org/, 2013.

\bibitem{roy2015inside}
A.~Roy, H.~Zeng, J.~Bagga, G.~Porter, and A.~C. Snoeren.
\newblock Inside the social network's (datacenter) network.
\newblock In {\em Proceedings of the 2015 ACM Conference on Special Interest
  Group on Data Communication}, pages 123--137, 2015.

\bibitem{sanches-whitley-tinos-tsp-px}
D.~S. Sanches, D.~Whitley, and R.~Tin{\'o}s.
\newblock Improving an exact solver for the traveling salesman problem using
  partition crossover.
\newblock In {\em Proceedings of Genetic and Evolutionary Computation
  Conference}, pages 337--344, 2017.

\bibitem{7066977}
S.~Schmid, C.~Avin, C.~Scheideler, M.~Borokhovich, B.~Haeupler, and Z.~Lotker.
\newblock {SplayNet}: Towards locally self-adjusting networks.
\newblock {\em IEEE/ACM Transactions on Networking}, 24(3):1421--1433, 2016.

\bibitem{singla2010proteus}
A.~Singla, A.~Singh, K.~Ramachandran, L.~Xu, and Y.~Zhang.
\newblock Proteus: a topology malleable data center network.
\newblock In {\em Proceedings of the 9th ACM SIGCOMM Workshop on Hot Topics in
  Networks}, pages 1--6, 2010.

\bibitem{sudholt-crossover-speeds-up-evco}
D.~Sudholt.
\newblock How crossover speeds up building block assembly in genetic
  algorithms.
\newblock {\em Evolutionary Computation}, (2):237--274, 2017.

\bibitem{Tarjan1984WorstcaseAO}
R.~E. Tarjan and J.~van Leeuwen.
\newblock Worst-case analysis of set union algorithms.
\newblock {\em J. ACM}, 31:245--281, 1984.

\bibitem{whitley-gpx-tsp-ec2020}
R.~Tin{\'o}s, D.~Whitley, and G.~Ochoa.
\newblock A new generalized partition crossover for the traveling salesman
  problem: Tunneling between local optima.
\newblock {\em Evolutionary Computation}, 28(2):255--288, 2020.

\bibitem{whitley-chicano-goldman}
L.~D. Whitley, F.~Chicano, and B.~W. Goldman.
\newblock Gray box optimization for {Mk} landscapes ({NK} landscapes and
  {MAX-kSAT}).
\newblock {\em Evolutionary Computation}, 24(3):491--519, 2016.

\bibitem{next-generation-ga-tutorial}
L.~D. Whitley, F.~Chicano, G.~Ochoa, A.~M. Sutton, and R.~Tin{\'o}s.
\newblock Next generation genetic algorithms.
\newblock In {\em Proceedings of Genetic and Evolutionary Computation
  Conference Companion}, pages 1113--1136, 2019.

\bibitem{wilcoxon-tests}
F.~Wilcoxon.
\newblock Individual comparisons by ranking methods.
\newblock {\em Biometrics Bulletin}, 1(6):80--83, 1945.

\end{thebibliography}

\end{document}